\documentclass[conference]{IEEEtran}
\usepackage{cite}
\usepackage{amsmath,amssymb,amsfonts,amsthm}
\usepackage{algorithmic}
\usepackage{graphicx}
\usepackage{textcomp}
\usepackage{xcolor}
\usepackage{url}
\def\BibTeX{{\rm B\kern-.05em{\sc i\kern-.025em b}\kern-.08em
    T\kern-.1667em\lower.7ex\hbox{E}\kern-.125emX}}
    
\DeclareMathOperator{\mean}{\mathbb{E}}
\newcommand{\be}{\begin{equation}}
\newcommand{\ee}{\end{equation}}
\newcommand{\prob}[1]{\hbox{Pr}\left\{#1\right\}}

\newcommand{\sset}[1]{\left\{#1\right\}}

\newcommand{\pd}{p_{\Delta}}

\newtheorem{thm}{Theorem}
\newtheorem{lemma}{Lemma}
\newtheorem{cor}{Corollary}

\begin{document}

\title{Matching Noisy Keys for Obfuscation\\
}

\author{\IEEEauthorblockN{Charlie Dickens}
\IEEEauthorblockA{\textit{Yahoo!} \\
charlie.dickens@yahooinc.com}
\and
\IEEEauthorblockN{Eric Bax}
\IEEEauthorblockA{\textit{Yahoo!} \\
ebax@yahooinc.com}
}

\maketitle

\begin{abstract}
Data sketching has emerged as a key infrastructure for large-scale data
analysis on streaming and distributed data.
Merging sketches enables efficient estimation of cardinalities and frequency histograms over distributed data.
However, merging sketches can require that each sketch stores hash codes for identifiers in different data sets or partitions, in order to perform effective matching. This can reveal identifiers during merging or across different data set or partition owners. 
This paper presents a framework to use noisy hash codes, with the noise level selected to obfuscate identifiers while allowing matching, with high probability.
We give probabilistic error bounds on simultaneous obfuscation and matching, concluding that this
is a viable approach.
\end{abstract}


\section{Introduction}

Large data sets are now used across a variety of domains, including biosciences, language processing, and online advertising. Despite the variety of applications, data analysis at scale often requires answering a common set of similar 
queries for different data sets \cite{goyal2011approximate,goyal2012sketch,rowe2019levee,kockan2020sketching}.
For example, estimating the number of unique identifiers in a dataset is called the \emph{count distinct} problem. Another canonical problem is the \emph{frequency estimation problem} -- estimating the distribution of item frequencies over items in a data set. For small data sets, these queries are simple to evaluate. 

However, streaming and/or distributed data sets are much more challenging. For them, there are a range of useful algorithms, known as \emph{sketching} algorithms or \emph{data sketches}, that make a single pass through the data, in parallel for distributed data sets, maintaining a small data structure in memory, called a \emph{sketch}. Sketches from multiple data sets can be merged to form a single sketch for the combined data set, and sketches maintain data sufficient to compute estimates that are accurate, with high confidence. Such mergeable sketches have seen widespread use in industry, for example in Google's BigQuery engine \cite{bigquery-hll} and in the open-source Apache DataSketches library \cite{asf-datasketches}. After processing, a merged sketch can be stored independently of the sketches from which it was built, and it can be used to answer multiple queries about the entire data set.

Conceptually, traditional mergeable sketches operate in a paradigm where the entity sketching the data is the same as the entity merging the sketches. However, with changing privacy regulations \cite{cummings2018role}, this assumption may be invalid. For example, a multinational company may own data warehouses in jurisdictions with differing privacy regulations which could prevent data-sharing across geographies even within the same organization. Rather than sending raw data, which could contravene regulation (and is slow at scale), organizations can send sketches to efficiently answer queries over data from multiple locations, if the data sketches can be shown to preserve privacy. More generally, entities holding data may differ from the entity wishing to answer a query. For example, an online advertising broker that runs campaigns across different vendors needs to estimate statistics for advertising campaigns over multiple vendors, yet the vendors want to avoid revealing sensitive statistics or data to the broker or other vendors. As a solution, vendors can send data sketches with sufficient noise to obfuscate each vendor's contribution to query results over the combined data. This has prompted novel multiparty computation and communication models for aggregating statistics, as presented in \cite{ghazi2019power,corrigan2017prio,ghazi2022multiparty}.

One challenge in developing such combined estimates is that merging many sketches relies on a key feature of hashing: that the same input always maps to the same output. 
Some sketches hash identifiers to produce hash codes, then transmit hash code-value pairs as the sketch, where the value is a count or some other number, for example Bottom-$k$--type sketches 
\cite{bar2002counting,cohen2007summarizing,dasgupta2015framework}. 
But using the hash codes as keys to compare allows parties with the hash function to hash identifiers and compare them to the hash codes transmitted in data sketches. If they match, then the party can infer, with high probability, that the identifier was contained in the data that generated the sketch. 
We address the issue of key obfuscation in this paper which is necessary in sketches that rely on key matching in the merge step such as 
\cite{ghazi2022multiparty,dasgupta2015framework}. 
Data privacy research more often addresses obfuscation of the values \cite{dwork2006differential,dwork2014algorithmic} but both are 
important in practice.

Our work addresses how to merge based on keys among \emph{untrusted parties}, to understand how sketches with non-matching hashes can be merged. In particular, we assume that entities sketching the data and merging the sketches do not trust each other. This invalidates the approach of \cite{dickens2022order} who assume that the hash function seed is a secret shared among those parties, and do not add noise to the hash codes. While we use noise to achieve obfuscation, another option is to encrypt the sketch prior to merging, but this is time- and resource-intensive \cite{ghazi2022multiparty}. Our contribution is to design a noisy hashing scheme that permits matching among noisy hashes and only has a small, controllable probability of revealing any true hash code shared by two non-trusting entities.

We achieve this by using a hash function with a seed shared among the parties, but then having each party flip hash code bits uniformly at random with some probability to add noise to each hash code. In general, distinct items tend to have hash codes that disagree in about half their bits. We show that, given a sufficient hash code bit-length and an appropriate probability of flipping each bit when adding noise, all hash codes that did not originally match are likely to maintain sufficient Hamming distance as noisy keys that considering clusters of close noisy keys as matching sets is likely to have few or no matching errors. Also given a sufficient hash code bit-length and bit-flipping probability, we show that there is sufficient noise that the center of a cluster of noisy keys is unlikely to match the hash code for the identifier represented by the noisy keys. Finally, we complement the theoretical analysis with a suite of characterizations showing how the performance varies depending on practical considerations such as the noise probability over bits in the hash code and the hash code length.

\subsection{Related Work.}
For distinct counting, there is a variety of work focusing on obfuscation and adding noise to obtain weak privacy \cite{von2019rrtxfm,tschorsch2013algorithm}. If sketches do \emph{not} have added noise and use only a single, public hash function, then a series of merge operations can identify the presence of an identifier with high confidence, which is clearly undesirable from a privacy perspective \cite{desfontaines2019cardinality}. If the hash function seed can be kept secret and cryptographic hash functions are used, then a wide class of practical sketches can be published under a strong differential privacy guarantee \cite{dickens2022order,dwork2006differential,dwork2014algorithmic}. When secretly seeded private noise is combined with a public hash function strong differential privacy guarantees are possible in some cases. These methods apply when sketch contributors and the merger all trust one another. 

Applying noise by using randomized response \cite{warner1965randomized} bit-flipping permits merging for the class of sketches \cite{hehir2023sketch} that store only bitmaps; including the Flajolet-Martin sketch \cite{flajolet1985probabilistic}. Other types of mergeable sketches that benefit from strong differential privacy guarantees include linear sketches for 
frequency estimation by adding noise to entries in the sketch \cite{zhao2022differentially, pagh2022improved}. 
This differs from simply adding noise to the estimator, as seen in \cite{mir2011pan,ghazi2019power,melis2015efficient}, which is not sufficient in the 
setting where merging sketches is required because merging requires access to the sketch data structure.

Using a single public hash function is not a serious problem if sketches keep only a single bit \cite{hehir2023sketch} or maintain additive increments in each bucket \cite{zhao2022differentially,pagh2022improved}, in which buckets can only increase in value by at most one based on an update from a new or previously seen identifier hashing to that bucket. This property, known as \emph{low sensitivity}, enables certain noise mechanisms to be applied to the sketches prior to merging.

In addition to low sensitivity, the bitmap-style sketches of \cite{hehir2023sketch} and frequency estimator sketches of \cite{zhao2022differentially,pagh2022improved} benefit from another property: their performance is not degraded by multiple identifiers hashing to the same bucket. In \cite{hehir2023sketch}, a bit is set from $0$ to $1$ if the corresponding bucket is occupied by at least one item in the stream. The estimator needs the total number of set bits and, crucially, does not depend on how many items have caused a particular bit to be set. Similarly, the counting arrays in \cite{zhao2022differentially,pagh2022improved} rely on the Count(Min) arrays of \cite{charikar2002finding,cormode2005improved}. These arrays have an independent hash function for each row to determine which rows in the array to increment. Even if distinct items land in the same bucket in one row of the array, they are unlikely to do so in further rows, enabling the estimate to compensate for them easily. So keys are not required in the data sketch. 

A mergeable sketch to estimate both frequency histograms\footnote{Note that frequency \emph{histogram} and frequency \emph{estimation} differ in the sketching literature.
The latter asks us to estimate the frequency of any identifier in the data stream while the former asks us to estimate the number of identifiers with a given count in the stream.} and distinct count in small space has been proposed by \cite{ghazi2022multiparty}.
The sketch maintains a one-dimensional array of counters whose buckets are incremented by all items landing in that location and a hash function that places exponentially more mass on low-index buckets. These buckets must be resolved from those that have only a single identifier land there for later estimation, unlike in the sketches listed previously. One issue with this approach is that low-indexed buckets with a high probability of being occupied have artificially large counts because multiple items are hashed to them. The second issue is that their collision-resolution uses homomorphic encryption to generate the keys that indicate which items are in each bucket. That requires a further hashing step per update and suffers from the computational overhead of encrypted arithmetic.

\section{Notation and Goals}
Suppose we have a set of data sources, and each has a set of values. 
For each value, the data source applies a hash function to get a hash code, then adds noise to the hash code by randomly selecting whether to flip the bit (change its value), independently for each bit in the hash code. 
The result is a noisy key.
Hash codes are generated by a public hash function
accessible to and computable by all sources.
The noisy key is generated using independent, private randomness for every source.

Next, all data sources transmit their noisy keys to a merge process. The merge process clusters the noisy keys into clusters having zero or one noisy keys from each data source. Informally, one goal is for the merge process to place all noisy keys resulting from the same value into a cluster with no other noisy keys, for each distinct value. 
If only one data source has a value, then the goal is to place its noisy key into a ``singleton cluster'' with only that noisy key. Another goal is for it to be difficult to infer the hash code for any noisy key, even if each data source supplies a noisy key for the same value. 

Before formalizing those goals, we introduce some notation. For any pair of values $v_a$ and $v_b$ from different data sources, let $a$ and $b$ be the hash codes, and let $a'$ and $b'$ be the noisy keys. All data sources use the same public hash function, so if $v_a = v_b$ then $a = b$. 
Also, assume that if $v_a \not= v_b$, then each bit position in $a$ and $b$ has probability (over hash functions) one half of having equal bit values. Our probabilities are over the choice of hash function as well as over random bit flips.

We use $d(a, b)$ and $d(a',b')$ to represent Hamming distance -- the number of bit positions with unequal bit values. Informally, we say that $d(a, b)$ is the number of disagreements between sequences $a$ and $b$. We refer to the number of bit positions that have equal bit values as the number of agreements. For example, $00011$ and $00110$ has Hamming distance 2, because it has two disagreements (in the third and fifth positions). 

Let $(a', b')$ be a pair of noisy keys with $a'$ from a different data source than $b'$. Let $t$ be a matching threshold, and declare $(a', b')$ a match if $d(a', b') < t$ and unmatched otherwise. Define two types of events we wish to avoid or minimize:
\begin{itemize}
\item $M$: $d(a', b') < t$ and $v_a \not= v_b$ -- noisy keys match for unequal values.
\item $U$: $d(a', b') \geq t$ and $v_a = v_b$ -- noisy keys for equal values are unmatched.
\end{itemize}

Given a set of noisy keys from the same value and the hash code for the value, define the median key to be bit sequence with each position's bit value equal to the majority bit value in that position over the set of noisy keys, with a tie broken in favor of the value in that position in the hash code. 
Specifically, if $a_i' = x$ and there are an equal number of bits
$x$ and $1-x$ collected, then the value $x$ is reported as the 
median in position $i$.
Obfuscation fails if the median key is the same as the hash code. For each noisy key, define event $R$ to be that the set of noisy keys that shares its hash code has median key the same as that hash code.

In the following sections, we will analyze probabilities and expectations of the undesirable events $M$, $U$, and $R$.  Then we examine key lengths and noise levels (in terms of bit-flipping probabilities) required for to control the probabilities of those events.

\section{Mismatches and Missed Matches}
\label{sec:main-theory}
Let $s$ be the number of data sources, and let $m_i$ be the number of noisy keys from data source $i$ for $1 \leq i \leq s$. Assume noisy keys from the same data source result from distinct values. We are concerned with matching among pairs of noisy keys from different data sources. Let $Q$ be the set of those pairs. Since $Q$ is all pairs of noisy keys except those from the same data source:
\be
|Q| = {{\sum_{i = 1}^s m_i}\choose{2}} - \sum_{i = 1}^s {{m_i}\choose{2}}. \label{size_q}
\ee

Let $n$ be the length of each hash code and each resulting noisy key. Let $p_f$ be the probability of flipping each bit while transforming a hash code into a noisy key. Let $\pd$ be the probability that bit-flipping alters one bit but not the other in some position in a pair of noisy keys. Since there are 2 choices for which bit to flip, and one must flip but not the other, 
\be
\pd = 2 p_f (1 - p_f).
\ee

\subsection{Mismatches}
Recall the binomial distribution's probability mass function (pmf):
\be
b(k, n, p) = {{n}\choose{k}} p^k (1 - p)^{n - k}.
\ee
And use $B(k_o, k_f, n, p)$ to denote the probability of a binomial random variable having value in $[k_o, k_f]$:
\be
B(k_o, k_f, n, p) = \sum_{k = k_o}^{k_f} {{n}\choose{k}} p^k (1 - p)^{n - k}.
\ee

Recall that $M$ is the event that for a pair $(a', b') \in Q$, noisy keys $a'$ and $b'$ match, but they come from different values. 
\begin{lemma}
Let
\be
p_M = \prob{d(a',b') < t | v_a \not= v_b}
\ee
be the probability of a match given different values. Then
\be
p_M = B\left(0, t - 1, n, \frac{1}{2}\right).
\ee
\end{lemma}

\begin{proof}
Recall that our probabilities are over selection of a hash function, and we assume that assigns each bit either 0 or 1 with probability $\frac{1}{2}$ for each value, independently over bits and between hash codes for $v_a \not= v_b$. If the bit-flip probability $p_f$ is independent of whether a bit has value 0 or 1, then, by symmetry, the probability of value 0 or 1 remains $\frac{1}{2}$ after random bit-flipping. To see this, the probability of 0 is the probability of starting with 0 and not flipping the bit plus the probability of starting with 1 and flipping the bit: $\frac{1}{2} (1 - p_f) + \frac{1}{2} p_f = \frac{1}{2}$. 
\end{proof}

\subsection{Missed Matches}
Recall that $U$ is the event that $(a', b') \in Q$ are unmatched and they come from the same value. 
\begin{lemma}
Let 
\be
p_U = \prob{d(a',b') \geq t | v_a = v_b}.
\ee
Then
\be
p_U = B(t, n, n, \pd).
\ee
\end{lemma}

\begin{proof}
Since $v_a = v_b$, their hash codes are equal: $a = b$. So adding noise must introduce at least $t$ disagreements to get $d(a',b') \geq t $. For each bit position, the probability of introducing disagreement is $\pd$.
\end{proof}

\subsection{Matching Errors Over All Pairs}
Since mismatches and missed matches are disjoint, the probability of either happening is at most the maximum of their probabilities:
\begin{lemma} \label{max_lemma}
Let $p_W$ be the probability that a pair of noisy keys $(a', b') \in Q$ has a mismatch or a missed match. Then
\be
p_W \leq \max(p_M, p_U).
\ee
\end{lemma}

\begin{proof}
Note that
\be
p_W = \prob{M \lor U}.
\ee
Since $M$ requires $v_a \not= v_b$ and $U$ requires $v_a = v_b$, they are disjoint events and $p_W = \prob{M} + \prob{U}$.
Note that
\begin{align}
\prob{U} 
&= \prob{d(a',b') \geq t \land v_a = v_b} \\
&= \prob{d(a',b') \geq t | v_a = v_b} \prob{v_a = v_b} \\
&= p_U \prob{v_a = v_b}. 
\end{align}
Similarly,
\be
\prob{M} = p_M \prob{v_a \not= v_b}.
\ee
So
\begin{align}
  p_W &=  \prob{M} + \prob{U} \\ 
  &= p_M \prob{v_a \not= v_b} + p_U \prob{v_a = v_b}
\end{align}
But $\prob{v_a \not= v_b} = 1 - \prob{v_a = v_b}$, so
\be
p_W = p_M (1- \prob{v_a = v_b}) + p_U \prob{v_a = v_b}.
\ee
Since this is a convex combination,
\be
p_W \leq \max(p_M, p_U).
\ee
\end{proof}

Consider the expected number of matching errors of either type:
\begin{thm} \label{mean_w}
Let $w$ be the number of mismatches and missed matches over all pairs in $Q$. Then
\begin{align}
    \mean w &= p_W |Q| \\
    &\leq \max(p_M, p_U) |Q| \\ 
    &= \max(p_M, p_U) \left[ {{\sum_{i = 1}^s m_i}\choose{2}} - \sum_{i = 1}^s {{m_i}\choose{2}} \right].
\end{align}
\end{thm}
\begin{proof}
By linearity of expectations, the expectation over $Q$ is the probability for each pair, $p_W$, times $|Q|$. For the second and third lines, apply Lemma \ref{max_lemma} for $p_W$ and Expression \ref{size_q} for $|Q|$. 
\end{proof}

Now consider the distribution of the number of matching errors:
\begin{cor} \label{markov}
\be
\forall h \geq 0: \prob{w \geq h} \leq \frac{\max(p_M, p_U) |Q|}{h}.
\ee
\end{cor}

\begin{proof}
Apply Markov's inequality to nonnegative random variable $w$, using Theorem \ref{mean_w} for the expectation.
\end{proof}

For the probability of no matching errors:

\begin{cor} \label{cor_zero}
\be
\prob{w = 0} \geq 1 -  \max(p_M, p_U) |Q|.
\ee
\end{cor}

\begin{proof}
Apply Corollary \ref{markov}, with $h = 1$.
\end{proof}

\section{Obfuscation}
Recall that the median key for a set of noisy keys and a hash code that was used to generate the noisy keys is defined as the bit sequence that agrees with the majority of the noisy keys in each position, and agrees with the original hash code in positions with a tie over the noisy keys. It is the ``best guess" for which hash code generated the noisy keys, with ties broken in favor of guessing correctly, to be conservative. Also recall that event $R$ occurs for a noisy key if and only if the set of noisy keys with the same hash code as the noisy key has median key equal to that hash code.

\begin{lemma}
Let $p_R(z)$ be the probability that, for a set of $z$ noisy keys that have the same hash code, for all $n$ positions, the median bit value is the same as the bit value for that position in the hash code. 
\be
p_R(z) = B\left(0, \left\lfloor \frac{z}{2} \right\rfloor, z, p_f\right)^n.
\ee
\end{lemma}

\begin{proof}
If there are at most $\lfloor \frac{z}{2} \rfloor$ bit flips, then the median bit value remains the same as the hash code bit value. This is independent among positions.
\end{proof}

For the expected number of noisy keys with hash codes ``revealed" by median keys: let $r$ be the number of noisy keys for which $R$ occurs.
\begin{thm} \label{ob_thm}
Let $r$ be the number of noisy keys for which $R$ occurs. Then
\be
\mean r \leq \max_{z \in \sset{1, \ldots, s}} p_R(z) \sum_{i = 1}^s m_i.
\ee
\end{thm}

\begin{proof}
Each noisy key may be in a set of from $1$ to $s$ noisy keys with the same hash code. So $p_R(z)$ is its probability of $R$ for some $z \in \sset{1, \ldots, s}$. Relying on linearity of expectation, we can sum those $p_R(z)$ values over the number of noisy keys.
\end{proof}

\section{Empirical Results}
\subsection{Numbers and Figures}
We now provide a series of plots illustrating the bounds derived in Section \ref{sec:main-theory}.
Figure \ref{n_req_plot} shows key lengths sufficient to ensure 95\% confidence that all pairwise matches are correct (based on Corollary \ref{cor_zero}), given 95\% confidence that each median of matched noisy keys is not equal to its original key -- so on average 95\% obfuscation (by Theorem \ref{ob_thm}). Each plotted value is for the worst-case arrangement of matching values over the $s$ data sources (generally, having full sets of $s$ matches) and for optimal bit-flip probability $p_f$ and matching threshold $t$. Noisy key lengths grow slowly in the number of keys per data source ($m$), and somewhat aggressively in the number of data sources contributing keys for matching ($s$). The key lengths are similar for even $s$ values and the next odd values, because we assume that a median noisy key with a tied number of zero and one values for a bit position is the same as the original key's bit value for the position, to be conservative.

\begin{figure}
\includegraphics[width=3.5in]{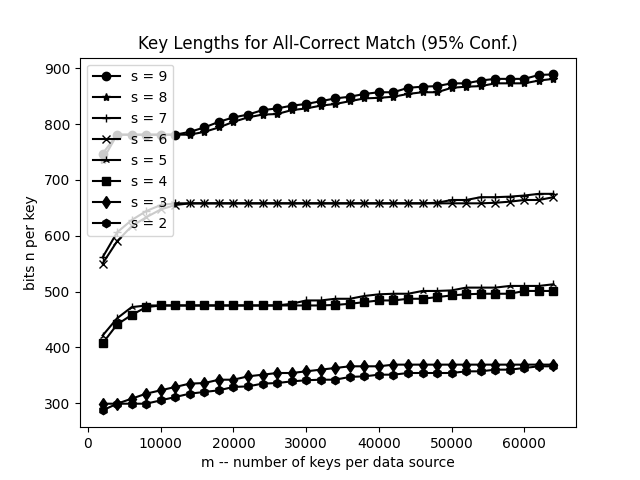} 
\caption{Number of bits per noisy key ($n$) sufficient to achieve at least 95\% confidence of no pairwise matching errors while maintaining a 95\% probability that the median of matched noisy keys does not match the original key. In general, the key length grows with the number of data sources $s$. Key lengths are nearly equal for each even $s$ and $s + 1$, because we define each median noisy key bit to have the bit value of the original key in case of a tie in bit values over the set of noisy keys.} \label{n_req_plot}
\end{figure}

Figures \ref{pr_plot} to \ref{q_size_plot} give a sense of the scales of probabilities of events $M$ (incorrectly matching different values), 
$U$ (incorrectly not matching two keys of the same value), and $R$ 
(the median key reveals the true hash code), and numbers of pairs of potential matches among noisy keys $|Q|$. Note that most of the plots are logarithmic on the $y$-axis. 

Figure \ref{pr_plot} presents probabilities of $R$ -- the event that the median key of a set of noisy keys from the same value may reveal the hash code that generated the key. 
That would allow anyone with the hash function and a value of interest to test whether the value hashes to the revealed hash code.
More noise (higher bit-flip probability $p_f$) makes revelation less likely, as does using more bits in the noisy keys, as it allows more opportunities for some bit position to be flipped in the majority of the noisy keys from the same value.
Inspecting the pairs of curves for $p_f=0.09, p_f=0.18$ and $p_f = 0.12, 
p_f=0.24$ we see that the
magnitude of the gradient increases by roughly a factor five.
This means that doubling the noise probability increases the exponential
rate parameter describing the probabilities by a factor of roughly five.
The plot is for 2 data sources.

\begin{figure}
\includegraphics[width=3.5in]{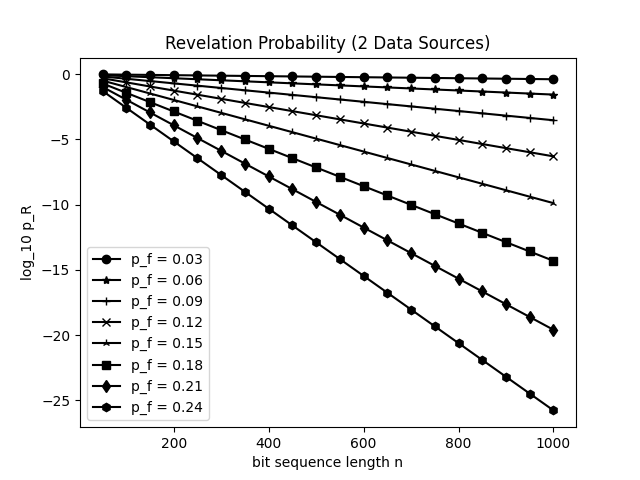} 
\caption{Upper bound on probability that the median key for a noisy keys set of matched noisy keys is equal to its hash code. More noise (higher bit-flip probability $p_f$) makes this possibility to identify a noisy key's hash code less likely, and more bits in the noisy key ($n$) make it exponentially less likely. (The values are for $s = 2$ data sources.)} \label{pr_plot}
\end{figure}

Figure \ref{pf_needed_plot} presents bit-flip probabilities $p_f$ required to achieve low revelation probabilities as the bit sequence length 
varies. 
The probabilities range from about a fifth to a third for $n = 50$ bits per noisy key (on the left) to about 5\% to 10\% for $1000$ bits per noisy key. The numbers presented are for 2 data sources -- they would increase with more data sources. Using less noise makes correct matching more likely.

\begin{figure}
\includegraphics[width=3.5in]{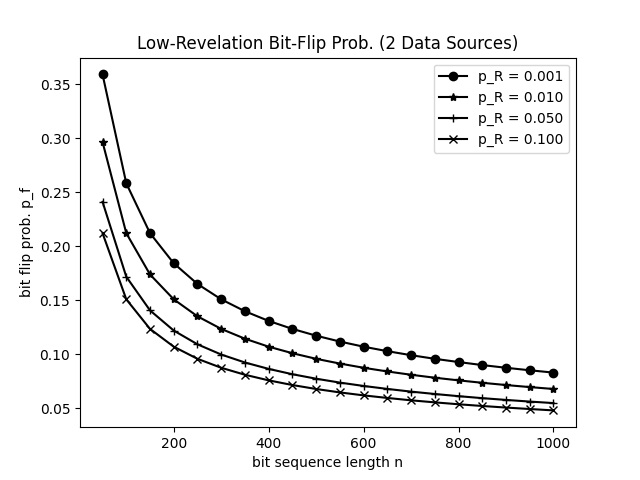} 
\caption{Minimum bit-flip probability to ensure the upper bound on revelation probability ($p_R$) is at most some selected values. Achieving a revelation probability bound less than 0.05 with fewer than 400 bits per noisy key requires bit-flip probabilities at least 0.05. For 50 bits per noisy key (the leftmost data points), a bit-flip probability of about 0.25 is required for a revelation probability bound of 0.05. (The values are for a pair of data sources.)} \label{pf_needed_plot}
\end{figure}

Figures \ref{pm_plot} and \ref{pu_plot} show probabilities of the 
incorrect matching events $M$ and $U$ for each pair of noisy keys from different data sources. 
Event $M$ is an incorrect match -- matching noisy keys that result from different values, put another way: failure to separate a pair of noisy keys that should not be matched. 
Event $U$ is failure to match a pair of noisy keys that should be matched. Failure to separate is independent of bit-flip probability $p_f$, but failure to match depends on it. The $p_f$ values used for Figure \ref{pu_plot} are the minimums necessary to bound the revelation probability $p_R$ by 5\% for 2 data sources. 

\begin{figure}
\includegraphics[width=3.5in]{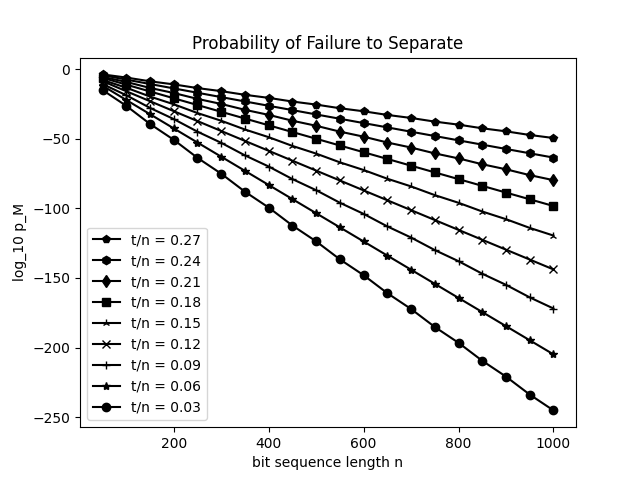} 
\caption{Probabilities $p_M$, for each pair of noisy keys, that they are matched, given that they should not be, that is, given that they are produced from different values. Each line is for a procedure that declares a match if fewer than the specified fraction of bits disagree. For each fraction, the probability decreases exponentially in $n$ -- the number of bits per noisy key.} \label{pm_plot}
\end{figure}

\begin{figure}
\includegraphics[width=3.5in]{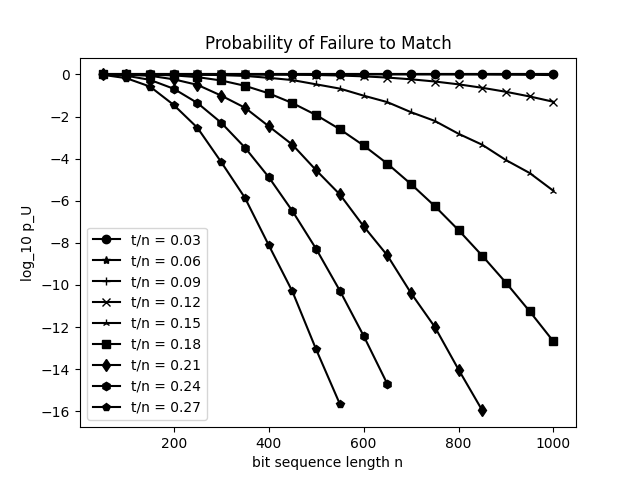} 
\caption{Probabilities $p_U$, for each pair of noisy keys, that they are left unmatched, given that they should be matched, since equal values produce them. Each line is for a procedure that declares a non-match if at least a specified fraction of bits disagree. For each number of bits in each noisy key ($n$), the probabilities shown are based on using the bit-flip probability $p_f$ that required for a revelation probability ($p_R$) upper bound of 0.05, with a pair of data sources. Decreases in $n$ are super-exponential because that $p_f$ value decreases with $n$, as shown in Figure \ref{pf_needed_plot}. Nonetheless, over 300 bits are needed per noisy key to make this probability less than one in a billion, even if noisy keys match if they agree on only three-quarters or more of their bits.} \label{pu_plot}
\end{figure}

Figure \ref{pw_plot} shows the probabilities of event $M \lor U$ -- either type of matching error -- for each pair of noisy keys from different data sources, given that bit-flip probability $p_f$ is large enough to ensure revelation probability $p_R \leq 0.05$. The results shown are for the optimal choices among those bit flip probabilities and matching threshold $t$ -- the choices that minimize the maximum of $p_M$ and $p_U$. Each line corresponds to a number of data sources; as that number grows, the probabilities of matching errors increase substantially, because $s$ data sources allows up to $s$ noisy keys for the same value, making it more likely that the median of those keys is the same as their (pre-noise) hash code, so forcing a higher bit-flip probability to achieve a low revelation probability. 

\begin{figure}
\includegraphics[width=3.5in]{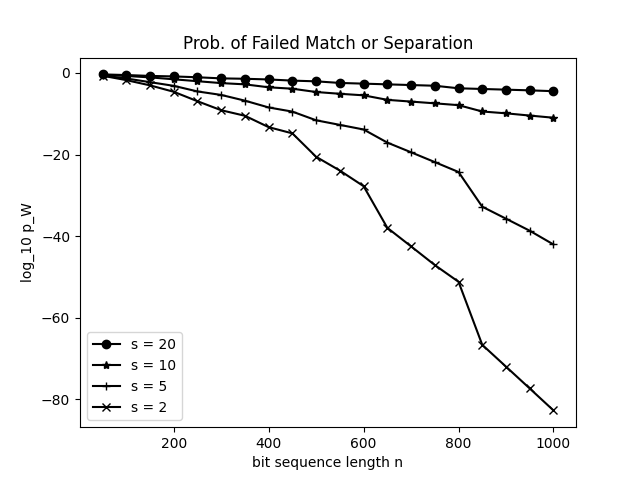} 
\caption{Probabilities $p_W$, for each pair of noisy keys, that they are incorrectly matched -- either matched though they have different values or left unmatched though they have the same values. Each point is for the optimal choice of bit-flip probability $p_f$ among those great enough to produce revelation probability $p_R \leq 0.05$ and the optimal choice of matching threshold $t \in \sset{1, ..., n}$. Each line shows results for a different number of data sources $s$.}
\label{pw_plot}
\end{figure}

Figure \ref{q_size_plot} shows how the number of pairs of noisy keys from different data sources grows as the number of keys per data source $m_i$ and the number of data sources $s$ increase. The growth is sub-exponential in the number of keys per data source, as it is O($ms^2$). In contrast, the values for $p_M$ and $p_U$ decrease exponentially and super-exponentially, respectively, in the number of bits per noisy key. This is promising as problem sizes increase, since the product of these numbers bounds the expected number of matching errors of any type (by Theorem \ref{mean_w}) and the probability of any matching error (by Corollary \ref{cor_zero}).

\begin{figure}
\includegraphics[width=3.5in]{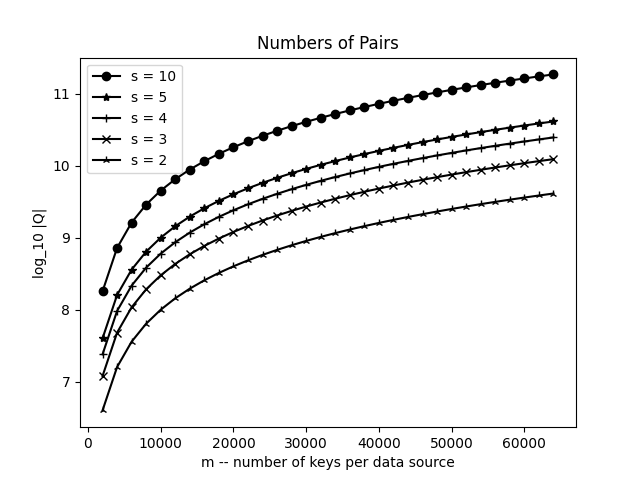} 
\caption{Numbers of pairs of noisy keys from different data sources. The base-10 logarithm of the expected number of matching errors for 2 data sources is at most the numbers plotted here plus the numbers on the line for $s = 2$ in Figure \ref{pw_plot}, according to Theorem \ref{mean_w}. As an example, for 2 data sources and $m_1 = m_2 = 64\,000$ noisy keys (the point on the right of the lowest line in this figure), the number of pairs $|Q|$ is between $10^9$ and $10^{10}$. For noisy key length $n = 400$ and $s = 2$, Figure \ref{pw_plot} shows that the error probability for each pair is less than $10^{-13}$. So that gives an expected number of matching errors on the order of $10^{-3}$, implying (via Corollary \ref{cor_zero}) about a 99.9\% probability of correct matching.} \label{q_size_plot}
\end{figure} 

\subsection{Discussion of Empirical Results}
The empirical results show that noisy key lengths required for correct matching with high probability while achieving 95\% obfuscation grow slowly in the number of keys per data source $s$ but aggressively in the number of contributors $s$ wishing to protect their data from each other. At the high end of the number of keys per source plotted (about $2^{16}$), the accuracy compares favorably 
with LiguidLegions \cite{ghazi2022multiparty}, which counts frequencies up to about $10^{10}$ using $2^{16}$ keys with relative error about $2.5\%$. (As far as we know, that is the only comparable method that merges noisy sketches for frequency histogram estimation.) 
Although LiquidLegions uses $256$ bits per key, their matching
scheme relies on cryptographic functions that are slow to 
evaluate.
Regarding the number of contributors, sketches that offer privacy consistently have performance degradation in the number of parties sharing data \cite{desfontaines2019cardinality, ghazi2022multiparty, hehir2023sketch}, with noticeable losses in accuracy as the number of contributors progresses from $2$. The amount of computation required for merging with privacy can also increase substantially, for example if a method's cryptographic communication protocol requires computation that is quadratic in the number of contributors \cite{ghazi2022multiparty}. 
Finally, the method of noisy keys presented in this paper can be applied to key-based methods such as Bottom-$k$ sketches \cite{beyer2009distinct,bar2002counting,cohen2007summarizing}
and the related Theta sketches \cite{dasgupta2015framework,asf-datasketches} that use a fixed number of keys, which is convenient since it allows implementors to determine sketch sizes a priori. 

\section{Conclusions and Future Work}
We have examined how random bit flips can obfuscate hash codes and still allow effective matching. We found that an upper bound on the probability that the median of noisy keys from the same hash code reveals the hash code -- the revelation probability -- implies a lower bound on the bit-flip probability required, and that bound grows with the number of data sources and shrinks with more bits per key. For matching error, selecting a threshold number of bit positions that must agree to declare a match mediates a tradeoff between two types of matching error -- matching noisy keys that have different values or failing to match those with the same value. The probabilities of both types of errors decrease exponentially as key length (in bits) increases. And bit-flip probability also plays a role in the probability of failing to match keys from equal values.

There are two ways to view error probabilities -- on a per-key or per-pair basis, or as a probability that no such errors occur over all keys or pairs. For revelation probabilities, we imposed per-key limits, meaning, for example, an upper bound of 5\% for revelation probability, which implies that, on average, fewer than 5\% of keys agree with their median keys. For matching error probabilities, we analyzed per-pair probabilities and probabilities over all pairs. For per-pair probabilities, the data for Figure \ref{pw_plot} show that the per-pair probability of either type of matching error is less than 5\% (for revelation probability at most 5\%) for 100-bit keys for up to 5 data sources, and 200-bit keys manage it for up to 10 data sources. However, achieving a low probability that no noisy key pair has a matching error of either type requires longer keys. For 2 data sources, achieving 95\% confidence that matching (or separation) is correct over all pairs requires about 300 bits for a few thousand noisy keys per data source, about 350 bits for about 10,000 keys per data source, and about 400 bits for 40 to 60,000 keys per source. 

Since one goal of sketching is to use data summaries that only require small memory footprints, using many bits per key may be a concern. 
In practice, some sketches use as few as 64 bits per key without noise and 
privacy 
\cite{asf-datasketches,dasgupta2015framework,heule2013hyperloglog}. 
However, it is known that more bits per key are needed for privacy 
deployments and both \cite{dickens2022order,ghazi2022multiparty} use $256$ bits per key.
Hence, there can be a cost (in bits) for the obfuscation achieved through noise. 
On the other hand, the numbers of bits detailed in this paper are total bits per key -- not additional bits, so any implementation using the number of bits detailed here or more would not need to add bits per key. 

Also, achieving high confidence that there are no pairwise matching or separation errors may be a conservative requirement for achieving correct matching, especially for greater numbers of data sources, which require longer noisy keys. With several data sources, it should be possible to recover from some matching errors. For example, with 5 data sources, if there are five noisy keys resulting from the same value, then there are ${{5}\choose{2}} = 10$ pairs among the noisy keys. If there are one or two errors, it may be possible to infer that the five should match on the strength of the evidence that eight or nine of the pairs match; and it is likely that the non-matching pairs would have Hamming distances close to the threshold value $t$.
More research is needed to understand how different levels of pairwise errors affect the accuracy of different matching methods. 



\bibliographystyle{IEEEtran} 
\bibliography{main}

\end{document}